\documentclass [journal,onecolumn,11pt]{IEEEtran}
\usepackage{amsfonts,amsmath,amssymb}
\usepackage{booktabs}
\usepackage{indentfirst, setspace}
\usepackage{url,float}
\usepackage{color}
\usepackage[a4paper,ignoreall]{geometry}
\usepackage{longtable}
\usepackage{array}
\usepackage{graphicx}
\usepackage[a4paper,ignoreall]{geometry}
\usepackage{multicol}
\usepackage{stfloats}
\usepackage{enumerate}
\usepackage{cite}
\usepackage{amsthm}
\usepackage{multirow}
\usepackage{amssymb}
\usepackage{subeqnarray}
\usepackage{cases}
\usepackage[square, comma, sort&compress, numbers]{natbib}

\def\qu#1 {\fbox {\footnote {\ }}\ \footnotetext { From Qu: {\color{red}#1}}}
\def\hqu#1 {}
\def\kq#1 {\fbox {\footnote {\ }}\ \footnotetext { From KangQuan: {\color{blue}#1}}}
\def\hkq#1 {}

\usepackage{extarrows}
\usepackage{titlesec}
\usepackage{amssymb}
\usepackage{amsmath,amsthm,amssymb,color}
\usepackage[pdftex]{hyperref}
\usepackage{graphicx}
\usepackage[varg]{txfonts}
\usepackage{multirow}
\date{}
\usepackage{url}

\newtheorem{corollary}{Corollary}

\newtheorem{lem}{Lemma}[section]
\newtheorem{thm}[lem]{Theorem}

\newtheorem{rmk}[lem]{Remark}

 \newcommand{\tabincell}[2]{\begin{tabular}{@{}#1@{}}#2\end{tabular}}

\begin{document}

\title{Constructing new  APN functions through relative trace functions}

\author{Lijing Zheng\thanks{
		\newline \indent L. Zheng is with the School of Mathematics and Physics, University of South China, Hengyang, Hunan, 421001, China,~(E-mail:~zhenglijing817@163.com).
		\newline \indent H. Kan is with the School of Computer Sciences, Fudan University,
		Shanghai, 200433, China, (E-mail: hbkan@fudan.edu.cn).
			\newline \indent Y. Li is with the Mathematics and Science College of
		Shanghai Normal University, Shanghai, 200234, China, (yanjlmath90@163.com).
		\newline \indent J. Peng is with the Mathematics and Science College of
		Shanghai Normal University, Shanghai, 200234, China, (jpeng@shnu.edu.cn).
\newline \indent D. Tang is with the School of Electronic Information and Electrical Engineering, Shanghai Jiao Tong University, Shanghai, 200240, China,  (dtang@foxmail.com).}, Haibin Kan, Yanjun Li, Jie Peng, Deng Tang }
\maketitle

\noindent {\bf Abstract:}  In 2020, Budaghyan, Helleseth and Kaleyski [IEEE TIT 66(11): 7081-7087, 2020] considered
an infinite family of quadrinomials over $\mathbb{F}_{2^{n}}$ of the form $x^3+a(x^{2^s+1})^{2^k}+bx^{3\cdot 2^m}+c(x^{2^{s+m}+2^m})^{2^k}$,
where $n=2m$ with $m$ odd. They proved that such kind of quadrinomials can provide new almost perfect nonlinear (APN) functions
when  $\gcd(3,m)=1$, $ k=0 $,  and $(s,a,b,c)=(m-2,\omega, \omega^2,1)$ or $((m-2)^{-1}~{\rm mod}~n,\omega, \omega^2,1)$ in which
$\omega\in\mathbb{F}_4\setminus \mathbb{F}_2$. By taking $a=\omega$ and $b=c=\omega^2$, we observe that such kind of quadrinomials can be
rewritten as $a {\rm Tr}^{n}_{m}(bx^3)+a^q{\rm Tr}^{n}_{m}(cx^{2^s+1})$,  where $q=2^m$ and $ {\rm Tr}^n_{m}(x)=x+x^{2^m} $ for $ n=2m$.
Inspired by the quadrinomials and our observation, in this paper we study a class of functions with the form  $f(x)=a{\rm Tr}^{n}_{m}(F(x))+a^q{\rm Tr}^{n}_{m}(G(x))$ 
and  determine the APN-ness of this new kind of functions, where $a \in \mathbb{F}_{2^n} $ such that $ a+a^q\neq 0$, and  both $F$ and $G$ are 
quadratic functions over $\mathbb{F}_{2^n}$.  We first obtain a characterization of the conditions for $f(x)$ such  that $f(x) $ is an APN function. 
With the help of this characterization,  we obtain an  infinite family of APN functions for $ n=2m $ with $m$ being an odd positive 
integer: $ f(x)=a{\rm Tr}^{n}_{m}(bx^3)+a^q{\rm Tr}^{n}_{m}(b^3x^9) $, where $ a\in \mathbb{F}_{2^n}$  such that $ a+a^q\neq 0 $ and $ b $ is a non-cube in $ \mathbb{F}_{2^n} $. 
We verify that the aforementioned APN quadrinomials are CCZ-inequivalent to any other known APN functions over $ \mathbb{F}_{2^{10}} $.  We also obtain two infinite families of APN functions: $ a{\rm Tr}^n_{m}(bx^{3})+ a^q{\rm Tr}^n_{m}(gx^{5}+ex^{4q+1}) $, where   $ b,~g,~e $ satisfy:
$i)$ $ b $ not a cube, $ g=1 $, $ e=\frac{1}{b^{2q-2}} $; or $ii)$ $ b $ not a cube, and $ g=e=b $. We can also find (at least) two new sporadic  instances of APN  functions over $ \mathbb{F}_{2^{10}} $ up to CCZ-equivalence.

\medskip

\section{Introduction}
Throughout this paper, we often identify the finite field $\mathbb{F}_{2^{n}}$ with $\mathbb{F}^{n}_{2}$ which is the $n$-dimensional vector space over $\mathbb{F}_{2}$. Any function $F: \mathbb{F}_{2^{n}}\rightarrow\mathbb{F}_{2^{m}}$ is called an {\it $(n,m)$-function} or vectorial Boolean functions if the values $ n $ and $ m $ are omitted. Vectorial Boolean functions are of critical importance in the field of symmetric cryptography, and the security of encryption algorithms heavily depends on the cryptographic properties of the vectorial Boolean functions. Researchers have proposed various properties to measure the resistance of a vectorial Boolean function to different kinds of cryptanalysis, including differential uniformity, nonlinearity, boomerang uniformity, algebraic degree, and so on. The lower the differential uniformity of a vectorial Boolean function, the better its security against differential cryptanalysis. In this paper, we mainly focus on the $ (n,n) $-functions. The differential uniformity of any such functions is at least 2, and the functions achieving this bound are called almost perfect nonlinear~(APN).

It is difficult to find new infinite families of APN functions up to CCZ-equivalence. Up to now, only 6 infinite families of APN monomials and 14 infinite families of APN polynomials are known, since the early 90's. On the other hand, in contrast to these facts, there are a lot of APN functions even over ``small'' field: for example, thousands of CCZ-inequivalent APN functions have been found over $ \mathbb{F}_{2^8}$ \cite{Yu-Wang-Li-2014}. Constructing new instances of infinite families is an area of deep heading research. We present Tables I and II including all currently known infinite families of APN functions. To Table II, we add the new function found with Theorem \ref{thm 3.1} in Section 3 below. We refer the readers to a recent nice work of Budaghyan et al. for more details on the classification of the known families of  APN functions   \cite{Budaghyan-Calderini-Villa-2020}.

\begin{table}[h]
	\centering
	\caption{Known infinite families of APN power functions over $ \mathbb{F}_{2^n} $}
	\label{table1} 
	\centering
	\begin{tabular}{|m{40pt}|m{98pt}|m{80pt}|m{80pt}|m{36pt}|}
		\hline
		Family & Exponent  & Conditions & Algebraic degree & Source \\  
	\hline
	Gold & $2^i+1$  &   ${\rm gcd}(i,n)=1 $& 2 &\cite{Gold-1968} \\
	\hline
	Kasami & $ 2^{2i}-2^i+1 $  & ${\rm gcd}(i,n)=1 $& $i+1$& \cite{Kasami-1971}\\
	\hline
	Welch & $ 2^t+3$  &  $ n=2t+1 $ & $3$& \cite{Dobbertin-1999} \\
	\hline
	Niho & \tabincell{l}{$ 2^t+2^{t/2}-1$, $t$ even\\ $2^t+2^{(3t+1)/2}-1$, $t$ odd } 
	&  $n=2t+1$ & \tabincell{l}{$t/2+1$\\$t+1$} &\cite{Dobbertin-1999-Niho} \\
		\hline
	Inverse & $ 2^{2t}-1$  &  $n=2t+1$&  $n-1$& \cite{Beth-Ding-1993, Nyberg-1994}\\
	\hline
	Dobbertin & $ 2^{4i}+2^{3i}+2^{2i}+2^i-1$  &  $ n=5i$ & $i+3$& \cite{Dobbertin-2001} \\
	\hline
	\end{tabular}
\end{table}

Throughout this paper, let $\omega\in \mathbb{F}_{4}
\backslash\{0,1\}.$ Very recently, Budaghyan, Helleseth, and Kaleyski introduced  an infinite family of quadrinomials over $ \mathbb{F}_{2^{n}} $ of the following form:
\begin{equation*}
g_{s}(x)=x^3+a(x^{2^s+1})^{2^k}+bx^{3\cdot 2^m}+c(x^{2^{s+m}+2^m})^{2^k},
\end{equation*}
where $ n=2m $. They showed that  this family can provide new infinite families of APN functions \cite{Budaghyan-Helleseth-Kaleyski-2020}.  More precisely, they showed that $ g_{s}(x) $ is a new APN function if $ k=0 $, $ (s,a,b,c)=(m-2,\omega, \omega^2,1)$, or $((m-2)^{-1}~{\rm mod}~n,\omega, \omega^2,1) $,  if $ m $ is odd with $ {\rm gcd}(3,m)=1 $. They also pointed out that when $ k\geq 1 $, $ g_{s}(x) $ can also be APN, however, CCZ-equivalent to some known ones. 

Let $ n=2m $ and $ q=2^m $. In this paper, our motivation is to find new infinite families of APN functions over $ \mathbb{F}_{2^n} $. We revisit the above-mentioned two infinite families of APN quadrionomials  obtained in \cite{Budaghyan-Helleseth-Kaleyski-2020}. Observing that for any odd positive integer $ s $, $ \omega^{2^s}=\omega^2$, the  APN functions for $ s=m-2 $, or $ (m-2)^{-1} {\rm mod}~n$ can be rewritten as $  g_{s}(x)=a {\rm Tr}^{n}_{m}(bx^3)+a^q{\rm Tr}^{n}_{m}(cx^{2^s+1})$, $ a=\omega$, $b=c=\omega^2 $. Here $ {\rm Tr}^n_{m}(x):=x+x^{2^m} $ for $ n=2m $. Inspired by  the quadrinomials and our  observation, let $ a\in \mathbb{F}_{2^n} $, we study a class of functions with the following form:
\begin{equation}\label{f(x)}
f(x)=a {\rm Tr}^{n}_{m}(F(x))+a^q{\rm Tr}^{n}_{m}(G(x)), ~a+a^q\neq 0, 
\end{equation}
where $ F $ and $G$ are quadratic functions with $ F(0)=G(0)=0 $. 

Based on the framework (\ref{f(x)}), we carefully choose quadratic functions $ F $ and $ G $ for finding APN functions. We mainly consider two kinds of functions in (\ref{f(x)}) by setting $ F $ and $ G $ as follows.

$ i) $ $ F(x)=bx^3 $, $ G(x)=cx^{2^s+1} $;

$ ii) $ $ F(x)=bx^{2^i+1}+cx^{2^{i+m}+1} $, $ G(x)=gx^{2^s+1}+ex^{2^{s+m}+1} $, where $ b, c, g, e\in\mathbb{F}_{2^n} $, and $ i, s $ are positive integers.

Let $ n=2m  $ with $ m $ odd. Let $ a\in \mathbb{F}_{2^n} $, and 
\begin{equation*}
f_{s}(x)=a {\rm Tr}^{n}_{m}(bx^3)+a^q{\rm Tr}^{n}_{m}(cx^{2^s+1}), ~a+a^q\neq 0.
\end{equation*} 
We can find two more exponents $ s=3$, or $m+2 $, and the corresponding  conditions on the coefficients  such that $ f_{s}(x) $ is an APN function over $ \mathbb{F}_{2^n} $. Code isomorphism tests~(see Sec. 2 below)~indicate that for the exponent   $s=3$, the APN function found with Theorem \ref{thm 3.1}:
\begin{equation*}
 f_{3}(x)=a {\rm Tr}^{n}_{m}(bx^3)+a^q{\rm Tr}^{n}_{m}(b^3x^{9}),
\end{equation*} 
where $ b $ is a non-cube, is new up to CCZ-equivalence over $ \mathbb{F}_{2^{10}} $. We can also discover more coefficients for these two exponents  $ s=m-2 $, and  $ (m-2)^{-1} {\rm mod}~n$ discovered by Budaghyan et al. such that $ f_{s}(x) $ is APN without the assumption that $ {\rm gcd}(3,m)=1 $. In this way, some new instances of APN functions  over $ \mathbb{F}_{2^{10}} $ and $ \mathbb{F}_{2^{14}} $ of  the  form $ f_{s}(x)  $ can also be found.

Let $ n=2m  $, $ q=2^m $, $ a\in \mathbb{F}_{2^n} $, and 
\begin{eqnarray*}
h_{i,s, b,c,g,e}(x)=a{\rm Tr}^n_{m}(bx^{2^i+1}+cx^{2^{i+m}+1})+ a^q{\rm Tr}^n_{m}(gx^{2^{s}+1}+ex^{2^{s+m}+1}),~a+a^q\neq 0. 
\end{eqnarray*} 	 
 We can find two infinite families of APN functions as follows, by letting $ i=1 $, $ s=2 $, $ c=0 $.
\begin{eqnarray*}
h_{1,2, b,0,g,e}(x)=a{\rm Tr}^n_{m}(bx^{3})+ a^q{\rm Tr}^n_{m}(gx^{5}+ex^{4q+1}) ,
\end{eqnarray*} 	 
where $ a\in \mathbb{F}_{2^n} $ such that $ a+a^q\neq 0 $, $ m $ is odd, and $ b,~g,~e $ satisfy: 
$i)$ $ b $ not cube, $ g=1 $, $ e=\frac{1}{b^{2q-2}} $; or $ii)$ $ b $ not cube in $ \mathbb{F}^{\ast}_{2^n} $, and $ g=e=b $.  By means of the code isomorphism test, we find that these two classes of  APN functions are CCZ-inequivalent to each other, however, CCZ-equivalent to some functions in family F12 of  Taniguchi over $ \mathbb{F}_{2^{10}}$. The critical technique needed in the proof is to forge links between the cube-ness of some certain elements and the number of   solutions to the  equation of the following form:  
\begin{eqnarray*}
	Ax^3+Bx^2+B^qx+A^q=0.
\end{eqnarray*}

The rest of the paper is organized as follows. Some basic definitions are given in Section 2. We characterize the condition  for $ f(x) $ with the form (\ref{f(x)}) such that $ f(x) $  is an APN function over $ \mathbb{F}_{2^{n}} $, $ n=2m $.  In Section 3, we investigate the APN property of the functions with the form (\ref{f(x)}) by letting $ F $, $ G $ are both Gold functions or both quadratic binomials. We can find a new infinite family of APN quadrinomials, and  generalize the two  infinite families of APN functions found by Budaghyan et al. in \cite{Budaghyan-Helleseth-Kaleyski-2020}. We can find two infinite families of APN hexanomials, which computationally proved  that they belong to family F12 over $ \mathbb{F}_{2^{10}} $. We can also find (at least) two new APN instances  over $ \mathbb{F}_{2^{10}} $. A few concluding remarks are given in Section 4. 

\section{Preliminaries}
Let $\mathbb{F}_{2^{n}}$ be the finite field consisting of $2^{n}$ elements, then the group of units of $\mathbb{F}_{2^{n}}$, denoted by $\mathbb{F}^{\ast}_{2^{n}}$, is a cyclic group of order $2^{n}-1$. Let $\alpha \in  \mathbb{F}_{2^n} .$ It is called a {\it cube} in $\mathbb{F}_{2^n} $, if $ \alpha=\beta^3 $ for some $\beta \in  \mathbb{F}_{2^n} $; otherwise, it is called a {\it non-cube}. Let $m$ and $n$ be two positive integers satisfying $m~|~n$, we use ${\rm Tr}^{n}_{m}(\cdot)$ to denote the {\it trace function} form $\mathbb{F}_{2^{n}}$ to $\mathbb{F}_{2^{m}}$, i.e., $ 	{\rm Tr}^{n}_{m}(x)=x+x^{2^m}+x^{2^{2m}}+\cdots+x^{2^{(n/m-1)m}}.$

Let $ f(x) $ be a function over $ \mathbb{F}_{2^n} $. Then it can be uniquely represented as $ f(x)=\sum^{2^n-1}_{i=0}a_{i}x^i $. This is the {\it univariate~representation} of $ f $. Let $ 0 \leq i\leq 2^n-1 $. The {\it binary~weight} of $ i $ is  $ w_{2}(i)=\sum^{n-1}_{s=0}i_{s}$, where $ i=\sum^{n-1}_{s=0}i_{s}2^s $, $ i_{s}\in \{0,1\} $. The {\it algebraic~degree} of $ f $, denoted by $ {\rm deg}(f) $, is the largest binary weight of an exponent $ i $ with $ a_{i}\neq 0 $ in the univariate representation of $ f $. Functions of algebraic degree one, and two are called {\it affine}, {\it quadratic}, respectively.

Given an $ (n,n) $-function $ F $, we denote by  $ \Delta_{F}(a,b) $ the number of solutions to the equation $ D_{a}F(x)=b $, where $ D_{a}F(x)=F(x)+F(x+a) $ is the \emph{derivative} of $ F $ in direction $ a\in \mathbb{F}_{2^n} $. $ F $ is called \emph{differentially $ \delta $-uniform} if the largest value of $ \Delta_{F}(a,b) $ equals to $\delta$, for every nonzero $ a $ and every $ b $. If  $ F $ is  differentially 2-uniform, we say that $ F $ is \emph{almost perfect nonlinear}~(APN).

Two $(n,m)$-functions $F$ and $G$ are called {\it extended affine equivalent} (EA-equivalent) if there exist some affine permutation $L_1$ over $\mathbb{F}_{2^n}$ and some affine permutation $L_2$ over $\mathbb{F}_{2^m}$, and some affine function $A$ such that $F=L_2\circ G\circ L_1+A$. They are called {\it Carlet-Charpin-Zinoviev equivalent} (CCZ-equivalent) if there exists some affine automorphism $L=(L_1, L_2)$ of $\mathbb{F}_{2^n}\times \mathbb{F}_{2^m}$, where $L_1: \mathbb{F}_{2^n}\times \mathbb{F}_{2^m}\rightarrow \mathbb{F}_{2^n}$ and $L_2: \mathbb{F}_{2^n}\times \mathbb{F}_{2^m}\rightarrow \mathbb{F}_{2^m}$ are affine functions, such that $y=G(x)$ if and only if $L_2(x, y)=F\circ L_1(x, y)$. It is well known that EA-equivalence is a special kind of CCZ-equivalence, and that  CCZ-equivalence preserves the differential uniformity \cite{CCZ}. Proving CCZ-inequivalence of functions can be very difficult in general, and this is resolved through code isomorphism. Let $ \alpha $ be the primitive element in $ \mathbb{F}_{2^n} $. Then two $ (n,n) $-functions functions $ F $ and $ G $ are CCZ-equivalent if and only if $\mathcal{C}_{F}$,  $\mathcal{C}_{G}$ are isomorphic \cite{Bracken-Byrne-Markin-McGuire-2008}, where $\mathcal{C}_{F}$ is the linear code corresponding to  $ F $ with the generating matrix as follows.
\begin{equation*}       
\mathcal{C}_{F}=\left(                
\begin{array}{cccc}   
1 & 1 & \cdots & 1\\ 
0 & \alpha & \cdots & \alpha^{2^n-1}\\ 
F(0) & F(\alpha) & \cdots & F(\alpha^{2^n-1})\\ 
\end{array}
\right)          
\end{equation*}

Let $ f $ be a quadratic function over $ \mathbb{F}_{2^n} $ with $f(0)=0 $. Denote 
\begin{equation*}
\Delta_{d,f}(x):=f(dx)+f(dx+d)+f(d).
\end{equation*}
Then it is well known that $ f $ is APN if and only if for every $ d\neq 0 $, $ \Delta_{d,f}(x)=0$ only has trivial solutions in $ x $, i.e., only $ x\in \mathbb{F}_{2}$ can be a solution to  $ \Delta_{d,f}(x)=0$.

In the following, we determine the APN-ness of the functions with the form (\ref{f(x)}).

\begin{lem}\label{fundamental-lemma} Let $ n=2m $, and $ q=2^m $. Let $ F $, $ G $ be  quadratic functions over $ \mathbb{F}_{2^n} $ satisfying that  $ F(0)=0 $, and $ G(0)=0 $. Let $ f(x)=a{\rm Tr}^{n}_{m}(F(x))+a^q{\rm Tr}^{n}_{m}(G(x)),$ where $ a\in \mathbb{F}_{2^n} $ such that $ a+a^q\neq 0 $. Then $ f(x)  $ is APN over $ \mathbb{F}_{2^n} $, if and only if  the following system
\begin{eqnarray}\label{fundamental}
\begin{cases}
	  \Delta_{d,F}(x) \in  \mathbb{F}_{2^m} &\\  \Delta_{d,G}(x) \in  \mathbb{F}_{2^m} &
\end{cases}
\end{eqnarray} 
only has  $ x=0, 1 $ as its solutions for any $ d \neq 0 \in\mathbb{F}_{2^n}	 $.
\end{lem}
\begin{proof}
 Since $ f(x)  $ is quadratic with $ f(0)=0 $, it is equivalent  to showing  that the following equation only has  $ x= 0,1 $ as its solutions for any $d\neq 0 $
\begin{equation}\label{f-1}
\Delta_{d,f}(x)=f(dx)+f(dx+d)+f(d)=0.
\end{equation}
We have 
\begin{equation}\label{f-2}
\Delta_{d,f}(x)=a{\rm Tr}^{n}_{m}(\Delta_{d,F}(x))+a^q{\rm Tr}^{n}_{m}(\Delta_{d,G}(x))=0.
\end{equation}	
In the following, we shall show that (\ref{f-2}) holds if and only if 
\begin{equation*}
{\rm Tr}^{n}_{m}(\Delta_{d,F}(x))={\rm Tr}^{n}_{m}(\Delta_{d,G}(x))=0.
\end{equation*}	
The sufficiency is clear. Let us show the necessity. 
	
Raising (\ref{f-2}) to its $ q $-th power, we have 
\begin{equation}\label{f-3}
a^q{\rm Tr}^{n}_{m}(\Delta_{d,F}(x))+a{\rm Tr}^{n}_{m}(\Delta_{d,G}(x))=0.
\end{equation}		
Adding (\ref{f-2}) and (\ref{f-3}), 
\begin{equation*}
(a+a^q){\rm Tr}^{n}_{m}(\Delta_{d,F}(x))+(a+a^q){\rm Tr}^{n}_{m}(\Delta_{d,G}(x))=0,
\end{equation*}		
which infers, since  $  a+a^q\neq 0$, that
\begin{equation}\label{f-4}
{\rm Tr}^{n}_{m}(\Delta_{d,F}(x))={\rm Tr}^{n}_{m}(\Delta_{d,G}(x)).	
\end{equation}		
Substituting (\ref{f-4}) into (\ref{f-2}), we can obtain
\begin{equation*}
{\rm Tr}^{n}_{m}(\Delta_{d,F}(x))={\rm Tr}^{n}_{m}(\Delta_{d,G}(x))=0, 	
\end{equation*}		
which is exactly the system (\ref{fundamental}). Therefore, $ f(x) $ is APN, if and only if the system  (\ref{fundamental}) only has trivial solutions $ x=0,1 $, for any $ d\neq 0 $.	
\end{proof}

\begin{table}[h]
	\centering
	\caption{Known infinite families of quadratic APN polynomials over $ \mathbb{F}_{2^n} $}
	\label{table2} 
	\centering
	\begin{tabular}{|m{20pt}|m{150pt}|m{170pt}|m{35pt}|}
	\hline
          ID & Functions  & Conditions & Source \\  
	\hline
	F1-F2 & $x^{2^s+1}+u^{2^k-1}x^{2^{ik}+2^{mk+s}}$  &   $ n=pk $, $ {\rm gcd}(k,p)={\rm gcd}(s,pk)=1 $, $ p\in \{3,4\} $, $ i=sk~{\rm mod}~p $, $ m=p-i $, $ n\geq 12 $, $ u $ primitive in $ \mathbb{F}^{\ast}_{2^n} $& \cite{Budaghyan-Carlet-Leander-2008}\\
	\hline
	F3 & $ sx^{q+1}+x^{2^i+1}+x^{q(2^i+1)}+dx^{2^iq+1}+d^qx^{2^i+q} $  & $ n=2m $, $ q=2^m $, $ {\rm gcd}(i,m)=1 $, $ d\in \mathbb{F}_{2^n} $, $ s\in \mathbb{F}_{2^n}\backslash\mathbb{F}_{2^m}$, $ X^{2^i+1}+dX^{2^i}+d^qX+1 $ has no solution $ x $ s.t. $ x^{q+1}=1 $& \cite{Budaghyan-Carlet-2008,Budaghyan-Calderini-Villa-2020}\\
	\hline
	F4 & $ x^3+a^{-1}{\rm Tr}^{n}_{1}(a^3x^9)$  &  $ a\neq 0 $ & \cite{Budaghyan-Carlet-Leander-2009}\\
			\hline
	F5  & $ x^3+a^{-1}{\rm Tr}^{n}_{3}(a^3x^9+a^6x^{18})$  &  $ 3~|~n $, $ a\neq 0 $ & \cite{Budaghyan-Carlet-Leander-2009-w}\\
	\hline
	F6  & $ x^3+a^{-1}{\rm Tr}^{n}_{3}(a^6x^{18}+a^{12}x^{36})$  &  $ 3~|~n $, $ a\neq 0 $ & \cite{Budaghyan-Carlet-Leander-2009-w}\\
	\hline
	F7-F9  & $ ux^{2^s+1}+u^{2^k}x^{2^{-k}+2^{k+s}}+vx^{2^{-k}+1}+\omega u^{2^k+1}x^{2^{s}+2^{k+s}}$  &  $n=3k$, $ {\rm  gcd}(k,3)={\rm  gcd}(s,3k)=1$, $v$, $\omega \in \mathbb{F}_{2^k}$, $v\omega \neq 1$, $ 3~|~(k+s) $, $ u $ primitive  in $ \mathbb{F}^{\ast}_{2^n} $& \cite{Bracken-Byrne-Markin-McGuire-2008,Bracken-Byrne-Markin-McGuire-2011}\\
	\hline
	F10 & $ cx^{q+1}+dx^{2^i+1}+d^qx^{q(2^i+1)}+\sum^{m-1}_{s=1}\gamma_{s}x^{2^s(q+1)} $  & $ n=2m $, $ q=2^m $, $ {\rm gcd}(i,m)=1 $, $ i $, $ m $ odd, $ \gamma_{s}\in \mathbb{F}_{q} $, $c \notin \mathbb{F}_{q}$, $ d $ not a cube & \cite{Bracken-Byrne-Markin-McGuire-2008}\\		
	\hline	
	F11 & $ (x+x^q)^{2^k+1}+u^{\prime}(ux+u^qx^q)^{(2^k+1)2^i}+u(x+x^q)(ux+u^qx^q) $  & $ n=2m $, $m\geq 2$ even, $ {\rm gcd}(k,m)=1 $, $ q=2^m $,  and $ i\geq 2 $ even, $ u $ primitive in $ \mathbb{F}^{\ast}_{2^n} $, $ u^{\prime}\in \mathbb{F}_{2^m} $ not a cube     & \cite{Zhou-Pott-2013}\\
	\hline
	F12 & $ u(u^qx+ux^q)(x+x^q)+(u^qx+ux^q)^{2^{2i}+2^{3i}}+\alpha(u^qx+ux^q)^{2^{2i}}(x+x^q)^{2^i}+\beta(x+x^q)^{2^{i}+1}$ & $ n=2m $, $q=2^m$, $ {\rm gcd}(i,m)=1 $, $ u $ primitive in $ \mathbb{F}^{\ast}_{2^n} $, $ \alpha $, $ \beta\in \mathbb{F}_{2^m} $, and $ X^{2^i+1}+\alpha X+\beta $ has no solution in $ \mathbb{F}_{2^m}  $& \cite{Taniguchi-2019}\\ 
	\hline
	F13 & $ L(x)^{2^i}x+L(x)x^{2^i} $ & $ n=km $, $ m\geq 2 $, $ {\rm gcd}(n,i)=1 $, $ L(x)=\sum^{k-1}_{j=0}a_{j}x^{2^{jm}} $ satisfies the conditions in Theorem 6.3 of \cite{Budaghyan-Calderini-Carlet-Coutter-Villa-2020} & \cite{Budaghyan-Calderini-Carlet-Coutter-Villa-2020}  \\ 
    \hline
	F14 & $ x^3+\omega x^{2^s+1}+\omega^2x^{3q}+x^{(2^s+1)q} $ & $ n=2m $, $q=2^m$, $ m$ odd, $ 3 \nmid m $, $\omega $ primitive in $ \mathbb{F}^{\ast}_{2^2} $, $ s=m-2 $, $ (m-2)^{-1}~{\rm mod}~n $  & \cite{Budaghyan-Helleseth-Kaleyski-2020}\\ 
	\hline
   F15 & $ a{\rm Tr}^{n}_{m}(bx^3)+a^q{\rm Tr}^{n}_{m}(b^3x^9) $ & $ n=2m $, $ m$ odd, $q=2^m$, $ a \notin \mathbb{F}_{q} $, $ b $ not a cube  & new \\ 
	\hline
	\end{tabular}
\end{table}

\section{Three infinite families of APN functions}

We want to  find new APN functions of the form (\ref{f(x)}). In the following two subsections, the functions   $ F $ and $ G $ were chosen very carefully  to satisfy the conditions characterized  in Lemma \ref{fundamental-lemma}. This will yield a new infinite family of APN quadrinomails, two infinite families of APN hexanomials,   and (at least) two sporadic APN functions CCZ-inequivalent to any other known APN functions over $ \mathbb{F}_{2^{10}} $.

\subsection{F, G are both of Gold type}

We need the following two lemmas, which will be used in the proof of Theorem \ref{thm 3.1}.

\begin{lem}\label{lemma} Let $ n=2m $ for $ m $ odd, $ q=2^m $. Suppose that for some $ c\in \mathbb{F}_{2^n} $ we have 
\begin{equation*}
c^3(c+c^2+c^4)^q\in \mathbb{F}_{2^m}.	
\end{equation*}		
Then $ c $ is a cube in $ \mathbb{F}_{2^n} $.	
\end{lem}
\begin{proof} Since $ {\rm gcd}(3, 2^m-1)=1$, any element of $ \mathbb{F}_{2^m}  $ is a cube. In the following, we assume that $ c\notin \mathbb{F}_{2^m} $. Noting that $ c^3(c+c^2+c^4)^q=c^{(q+1)+2}+c^{2(q+1)+1}+c^{3(q+1)+q}  $, we have  $ c^{q+1}(c+c^q)^2+c^{2(q+1)}(c+c^q)+c^{3(q+1)}(c+c^q)=0 $ by the assumption that $ c^3(c+c^2+c^4)^q \in \mathbb{F}_{2^m} $. Since $ c+c^q\neq 0 $, we have $ c^{q+1}(c+c^q)+c^{2(q+1)}+c^{3(q+1)}=0 $, and hence
$ c+c^q=c^{q+1}+c^{2(q+1)} $. Note that any nonzero element $ c$ of $ \mathbb{F}_{2^n} $ has a unique polar decomposition of the form $ c=vk $, where $ k^{q+1}=1 $, and $ v^{q-1}=1 $. Substituting $ c=vk $ into  $ c+c^q=c^{q+1}+c^{2(q+1)} $, we have $ k+k^{-1}=v+v^3 $. By assumption that $ c\notin \mathbb{F}_{2^m} $, we have $ k\neq 1 $. Then according to \cite[Theorem 7]{KimMesnager2020}, we have that $ k $ is a cube in $ U:=\{x\in \mathbb{F}_{2^n}~|~x^{q+1}=1\}  $. Therefore, $ c=vk $ is a cube in $ \mathbb{F}_{2^n} $. \end{proof}

Let $ s $ be a positive integer with $ {\rm gcd}(s,n)=1 $. Let $ x\in \mathbb{F}_{2^n} $. It is clear that $ x+x^{2^s} \neq 0 $, if and only if $ x\neq 0,1 $. We have the following lemma.

\begin{lem}\label{lemma-2} Let $ n=2m $ for $ m $ odd with $ {\rm gcd}(3,m)=1 $. Let s be a positive integer such that $ 3s \equiv 1 ~{\rm mod}~n $. Suppose that for some $ x\in \mathbb{F}_{2^n}\backslash
	\{0,1\} $,  we have 
\begin{equation*}
\frac{x+x^2}{(x+x^{2^s})^{2^{2s}-2^s+1}} \in \mathbb{F}_{2^m}.
\end{equation*}		
Then $x+x^{2^s} $ is a cube.	
\end{lem}
\begin{proof}
Let $d=x+x^{2^s} $. Then $ d\neq 0 $, since $ x\neq 0, 1 $, and $ {\rm gcd}(s,n)=1 $. We can express $ x+x^2=d+d^{2^s}+d^{2^{2s}}$. Then 
\begin{equation*}
\frac{x+x^2}{(x+x^{2^s})^{2^{2s}-2^s+1}}=\frac{d+d^{2^s}+d^{2^{2s}}}{d^{2^{2s}-2^s+1}}=d^{-2^s(2^s-1)}+d^{-(2^s-1)^2}+d^{2^s-1}=A^{-2^s}+A^{-2^s+1}+A,
\end{equation*}			
where $ A=d^{2^s-1} $. Then the condition of this lemma is equivalent to that $ A^{-2^s}+A^{-2^s+1}+A+1\in \mathbb{F}_{2^m}, $	which is exaclty 
\begin{equation*}
\frac{(A+1)^{2^s+1}}{A^{2^s}}\in \mathbb{F}_{2^m}.
\end{equation*}		
If $ A=1 $, i.e., $ d^{2^s-1}=1 $, then $ d=1 $, and hence $x+x^{2^s}=1$ is a cube. In fact, since $ {\rm gcd}(2^s-1,2^n-1)=1 $, $ g(x)=x^{2^s-1} $ is a permutation of $ \mathbb{F}_{2^n} $. Then by $g(d)=g(1)=1 $, we have $ d=1 $. If $ A\neq 1 $, then there exists some $\alpha\in \mathbb{F}^{\ast}_{2^m}$ such that $ A^{2^s}=(A+1)^{2^s+1}\alpha $. Since $ s $ is odd, $ 3 ~|~2^s+1 $, we have $ A^{2^s+1}\alpha $ is a cube, and hence   $ A^{2^s} $ is a cube, that is, $ A $ is a cube. However, note that $ {\rm gcd}(3,2^s-1)=1 $, we have that $ d $  is a cube, when $A=d^{2^s-1} $ is.
\end{proof}

In the following theorem, we investigate the APN property of the functions with the form (\ref{f(x)}) by letting  $ F(x)=bx^3 $, and $ G(x)=cx^{2^s+1} $. This allows us to find a new infinite family of APN quadrinomials $f(x)=a{\rm Tr}^{n}_{m}(bx^3)+a^q{\rm Tr}^{n}_{m}(b^3x^9) $, where $ b $ is a non-cube in $ \mathbb{F}_{2^n} $.

\begin{thm}\label{thm 3.1} Let $n=2m$ with $ m \geq 1 $ odd, and $ q=2^m $. Let $ a\in \mathbb{F}_{2^n} $, and $ f_{s}(x)=a{\rm Tr}^n_{m}(bx^3)+ a^q{\rm Tr}^n_{m}(cx^{2^{s}+1})$ with $ a\notin \mathbb{F}_{q} $, $ bc\neq 0 $, $ s $   odd. Then $ f_{s}(x) $ is APN over $ \mathbb{F}_{2^n} $, if $s, b, c$ satisfy  the following
	
i) $ s=m-2 $, $ b $ not a cube, $  \frac{c^4}{b}\in \mathbb{F}_{2^m} $; or

ii) $ s=(m-2)^{-1}~{\rm mod}~n$, $ b $  not a  cube, $  \frac{c^{2^s-1}}{b^{2^{2s}}}\in \mathbb{F}_{2^m} $; or

iii) $ s=3 $, $ b $ not a cube, $  \frac{c}{b^3}\in \mathbb{F}_{2^m} $; or

iv) ${\rm  gcd}(3,m)=1$, $ 3s \equiv 1 {\rm ~mod~} n$, $ b $ not a cube,  $  \frac{c}{b^{2^{2s}-2^s+1}}\in \mathbb{F}_{2^m} $; or

v)  $ s=m$, $ b $ not a cube,  $ c \notin \mathbb{F}_{2^m};$ or

vi) $ s=m+2 $, $ b $ not a cube, $  bc\in \mathbb{F}_{2^m} $; or

vii) $ s=n-1$, $\frac{c^2}{b}\notin \mathbb{F}_{2^m} $.

\end{thm}
\begin{proof} 
Let $ F(x)=bx^3 $, $ G(x)=cx^{2^s+1} $. Then 	
\begin{equation*}
\Delta_{d,F}(x)=d^3b(x^2+x), ~{\rm and }~ \Delta_{d,G}(x)=d^{2^s+1}c(x^{2^s}+x).	
\end{equation*}		
 According to Lemma \ref{fundamental-lemma}, proving $ f_{s}(x) $ is an APN function over $ \mathbb{F}_{2^n} $ is equivalent to showing that  the system: $ \Delta_{d,F}(x) \in \mathbb{F}_{2^m}$, and $ \Delta_{d,G}(x)\in \mathbb{F}_{2^m} $ can only has trivial solutions $ x=0, 1 $ for any $ d\neq 0 $. Assume, to the contrary, that $ f_{s}(x) $ is not an APN function, when $ s, b, c $ satisfy   the conditions of one item in this theorem. Then the following system
\begin{eqnarray}\label{3}
\begin{cases}
	d^3b(x^2+x)=\alpha ,&\\ d^{2^s+1}c(x^{2^s}+x)=\beta. &
	\end{cases}
\end{eqnarray} 
has a non-trivial solution $ x\notin \mathbb{F}_{2} $ for some $ d\neq 0 $, where  $ \alpha,~\beta \in \mathbb{F}_{2^m} $ with $ \alpha \neq 0 $.

 Since $ m $ is odd, $ {\rm gcd}(3,2^m-1)=1 $, we have that  $ \alpha=e^3 $ for some $ e\in \mathbb{F}^{\ast}_{2^n} .$ Dividing both sides of the first equation in (\ref{3}) by $ e^3 $, we obtain that $ (d/e)^3b(x^2+x)=1 $. Dividing  both sides of the second equation in (\ref{3}) by $ e^{2^s+1} $, we have $ (d/e)^{2^s+1}c(x^{2^s}+x)=\beta e^{-(2^{s}+1)} $. Since $ s $ is odd, we have $ 3~|~2^s+1 $, and $ e^{2^s+1} \in \mathbb{F}_{2^m}.$  Therefore, the system (\ref{3}) has a non-trivial solution $ x\notin \{0,1\} $ if and only if the system 
\begin{eqnarray}\label{4}
\begin{cases}
d^3b(x^2+x)=1 ,&\\ d^{2^s+1}c(x^{2^s}+x)=\beta. &
\end{cases}
\end{eqnarray} 
has a solution for some $ d\in \mathbb{F}^{\ast}_{2^n}$ and $\beta \in \mathbb{F}_{2^m}. $

$ i) $ $ s=m-2 $, $ b $ is a non-cube in $ \mathbb{F}_{2^n} $ and $  \frac{c^4}{b}\in \mathbb{F}^{\ast}_{2^m} $.

 Raising the second equation in (\ref{4}) to its fourth power, we have $ d^{q+4}c^4(x^q+x^4)=\beta^4 $. From the first equation, we have $ d^3=\frac{1}{b(x^2+x)} $. Substituting this relation into the previous equation, we have $ d^{q+1}  \frac{c^4}{b} \frac{x^q+x^4}{x^2+x}\in \mathbb{F}_{2^m} $. Since $ d^{q+1}\in \mathbb{F}^{\ast}_{2^m}$, and $ \frac{c^4}{b}\in \mathbb{F}^{\ast}_{2^m} $ by assumption, we have $ \frac{x^q+x^4}{x+x^2}\in \mathbb{F}_{2^m} $. By  \cite[Lemma  1]{Budaghyan-Helleseth-Kaleyski-2020}, we have $ x+x^2 $ is a cube in $ \mathbb{F}_{2^n} $, and hence  $ b $ is a cube by  $ d^3b(x^2+x)=1 $, a contradiction to the assumption that $ b $ is a non-cube.

$ ii) $ $ s=(m-2)^{-1}~{\rm mod}~n  $, $ b $ is a non-cube in $ \mathbb{F}^{\ast}_{2^n} $ and $  \frac{c^{2^s-1}}{b^{2^{2s}}}\in \mathbb{F}^{\ast}_{2^m} $.

It can be seen from  the proof of Theorem 2 in \cite{Budaghyan-Helleseth-Kaleyski-2020} that the critical conditions ensuring the APN-ness of this $ f_{s}(x) $ are exactly that $ b $ is a non-cube in $ \mathbb{F}_{2^n} $ and $  \frac{c^{2^s-1}}{b^{2^{2s}}}\in \mathbb{F}^{\ast}_{2^m} $. We invite the readers to check it, and we omit the arguments here. 

$ iii) $ $ s=3 $, $ b $ is a non-cube in $ \mathbb{F}_{2^n} $ and $  \frac{c}{b^3}\in \mathbb{F}^{\ast}_{2^m} .$

It can be seen that in this case (\ref{4}) becomes 
\begin{eqnarray*}
\begin{cases}
d^3b(x^2+x)=1 ,&\\ d^{9}c(x^{8}+x)=\beta. &
\end{cases}
\end{eqnarray*} 
Substituting $ d^3=\frac{1}{b(x+x^2)} $ into the second equation of the above system, we have 
\begin{eqnarray*}
\frac{c}{b^3}\cdot \frac{x+x^8}{(x+x^2)^3}=\beta, 
\end{eqnarray*} 
which infers that $ \frac{x+x^8}{(x+x^2)^3}\in \mathbb{F}_{2^m} $, since  $ \frac{c}{b^3}\in \mathbb{F}^{\ast}_{2^m} $ by assumption. It implies that $ (x+x^2)^3(x+x^8)^q\in \mathbb{F}_{2^m} $. Denoting $ e=x+x^2 $, we have $ x+x^8=e+e^2+e^4 $, and  hence $ e^3(e+e^2+e^4)^q\in \mathbb{F}_{2^m} $. Now, according to Lemma \ref{lemma}, $  e=x+x^2 $ is a cube. Then $ b $ is a cube by $ d^3b(x+x^2)=1 $, which contradicts to the assumption that $ b $ is a non-cube.

$ iv)  $ ${\rm  gcd}(3,m)=1$, $ 3s \equiv 1 {\rm ~mod~} n$, $ b $ is a non-cube in $ \mathbb{F}_{2^n} $ and $  \frac{c^{2^{2s}-2^s+1}}{b}\in \mathbb{F}^{\ast}_{2^m} $.

Since $ {\rm gcd}(2^s-1,2^n-1)=2^{{\rm gcd}(s,n)}-1=1 $, we have that $ x+x^{2^s}\neq 0 $, when $ x\neq 0,1 $. Then (\ref{4}) becomes 
\begin{eqnarray*}
	\begin{cases}
		d^{2^{3s}+1}b(x+x^2)=1 ,&\\ d^{2^s+1}c(x+x^{2^s})=\beta, &
	\end{cases}
\end{eqnarray*} 
where $ \beta\in \mathbb{F}_{2^m} $ with $ \beta\neq 0 $, since $ x+x^{2^s}\neq 0 $. By the second equation, we have $ d^{2^s+1}=\frac{\beta}{c(x+x^{2^s})} $. Substituting this relation into the first equation, noting that $ 2^{3s}+1=(2^s+1)(2^{2s}-2^s+1) $, we have 
\begin{eqnarray*}
\frac{b}{c^{2^{2s}-2^s+1}}\cdot \frac{x+x^2}{(x+x^{2^s})^{2^{2s}-2^s+1}}\in \mathbb{F}_{2^m},
\end{eqnarray*} 
which infers, since $ \frac{b}{c^{2^{2s}-2^s+1}} \in \mathbb{F}^{\ast}_{2^m}$ by assumption, that 
\begin{eqnarray}\label{5}
\frac{x+x^2}{(x+x^{2^s})^{2^{2s}-2^s+1}}\in \mathbb{F}^{\ast}_{2^m}.
\end{eqnarray} 
Now, by the assumption that $ b $ is a non-cube in $ \mathbb{F}_{2^n} $ and $  \frac{c^{2^{2s}-2^s+1}}{b}\in \mathbb{F}^{\ast}_{2^m} $, we have that $ c $ is a non-cube. On the other hand, by (\ref{5}) and Lemma \ref{lemma-2}, we have that $ x+x^{2^s} $ is a cube, which infers that $ c $ is a cube from the second equation $ d^{2^s+1}c(x+x^{2^s})=\beta $ of the above system, a contradiction.

$ v) $  $ s=m$, $ b $ is a non-cube in $ \mathbb{F}_{2^n} $, and  $ c \notin \mathbb{F}_{2^m}.$

It can be seen that (\ref{4}) becomes 
\begin{eqnarray*}
	\begin{cases}
		d^{3}b(x+x^2)=1 ,&\\ d^{2^m+1}c(x+x^{2^m})=\beta, &
	\end{cases}
\end{eqnarray*} 
where $ \beta\in \mathbb{F}_{2^m} $. Since 
 $ c \notin \mathbb{F}_{2^m}$, and $ d^{2^m+1}\in \mathbb{F}^{\ast}_{2^m},$ $ x+x^{2^m}\in \mathbb{F}_{2^m} $ for any $ d\neq 0 $, $ x\in \mathbb{F}_{2^n}$, by the second equation, we have $ \beta $ must equal to  zero, which infers that $ x\in \mathbb{F}_{2^m} $. Then by the fact that  any element of $ \mathbb{F}_{2^m} $ is a cube, we have  $ d^3(x+x^2) $ is a cube in $ \mathbb{F}^{\ast}_{2^n} $,   which implies that $ b $ is a cube in $ \mathbb{F}^{\ast}_{2^n} $, a contradiction to the assumption that $ b $ is a non-cube.

$ vi) $ $ s=m+2 $, $ b $ is a non-cube in $ \mathbb{F}_{2^n} $ and $  bc\in \mathbb{F}^{\ast}_{2^m} $.
It can be seen (\ref{4}) becomes 
\begin{eqnarray*}
	\begin{cases}
		d^{3}b(x+x^2)=1 ,&\\ d^{4(q+1)-3}c(x+x^{4q})=\beta, &
	\end{cases}
\end{eqnarray*} 
where $ \beta\in \mathbb{F}_{2^m} $ with $ \beta\neq 0 $ since $ x+x^{4q}\neq 0 $ when $ x\neq 0, 1$. Since $ d^{3}b(x+x^2)=1 $, we have  $ d^{3}=\frac{1}{b(x+x^2)} $. Substituting this relation into the second equation, we have 
\begin{eqnarray*}
 d^{4(q+1)}bc(x+x^2)(x+x^{4q})=\beta.
\end{eqnarray*} 
Then by the assumption that $ bc\in \mathbb{F}^{\ast}_{2^m} $, we have  $ (x+x^2)(x+x^{4q})\in \mathbb{F}_{2^m} $. According to \cite[Lemma 1]{Budaghyan-Helleseth-Kaleyski-2020}, we have $ x+x^2\neq 0$ is a cube, which infers that $ b $ is a cube by  $ d^{3}b(x+x^2)=1 $, a contradiction to the assumption that $ b $ is a non-cube.

$ vii) $ $ s=n-1$, $\frac{c^2}{b}\notin \mathbb{F}_{2^m} $.

Since $ {\rm gcd}(2^s-1,2^n-1)=2^{{\rm gcd}(s,n)}-1=1 $, we have that $ x+x^{2^s}\neq 0 $, if $ x\neq 0,1 $.
It can be seen that (\ref{4}) becomes 
\begin{eqnarray*}
	\begin{cases}
		d^{3}b(x+x^2)=1 ,&\\ d^{2^s+1}c(x+x^{2^s})=\beta, &
	\end{cases}
\end{eqnarray*} 
where $ \beta\in \mathbb{F}_{2^m}$ with $ \beta \neq 0 $. Squaring  the second equation, we have 
$ d^{3}c^2(x+x^2)=\beta^2 $. Comparing with the first equation, we have $ \frac{c^2}{b}=\beta^2 \in \mathbb{F}_{2^m}$, which contradicts with the assumption that $ \frac{c^2}{b}\notin \mathbb{F}_{2^m} .$
\end{proof}
\begin{rmk} Code isomorphism tests described in Section 2 suggest that all the polynomials from the same item of Theorem \ref{thm 3.1} are all CCZ-equivalent;  the APN function  $x^3+\omega x^{2^s+1}+\omega^2 x^{3q}+x^{(2^s+1)q} $ discovered in \cite{Budaghyan-Helleseth-Kaleyski-2020} is CCZ-equivalent to all the functions in i), ii), respectively, for $ s=m-2 $, and $ s=(m-2)^{-1}~{\rm mod}~n $, if $ {\rm gcd}(3,m)=1 $; the polynomials $ f_{s}(x) $ for  $ s=m+2 $ in vi) are equivalent to the ones for $ s=m-2 $ in i);  the polynomials  $ f_{s}(x) $ for  $ s=m $  in v)  are equivalent to some functions in family F10 from Table \ref{table2}, see also the arguments in Remark \ref{rmk-f_{m}} below; the polynomial  $ f_{s}(x) $ for  $ s=n-1 $  in vii) is CCZ-equivalent to $ x^3 $. 
	
	The remaining value of $ s=3$ in iii) yields  APN quadrinomials $ f_{3}(x)$, which are CCZ-inequivalent to any currently known APN function over $ \mathbb{F}_{2^{10}} $. By the arguments above that all the polynomials in the same item are all CCZ-equivalent, we   only take a representative of iii). We let $ f_{3}(x)=\omega {\rm Tr}^{n}_{m}(bx^3)+\omega^2{\rm Tr}^{n}_{m}(b^3x^9) $, where $ b $ is  a non-cube, $ \omega \in \mathbb{F}_{2^2} \backslash \mathbb{F}_{2} $. We use this $ f_{3}(x) $ to compare against representatives from all the known infinite families including $ f_{s}(x)$, $s=m-2$, $(m-2)^{-1}~{\rm mod}~n $ in i), ii)   which are essentially due to Budaghyan, Helleseth, and Kaleyski~(\cite{Budaghyan-Helleseth-Kaleyski-2020}). Note that, Budaghyan et al. had  presented a table  listing all the representatives, except family F12, of all the known CCZ-inequivalent APN functions over $\mathbb{F}_{2^{10}}$, see Table III of \cite{Budaghyan-Helleseth-Kaleyski-2020}. To complete the work of code isomorphism test, we have to find all the representatives of F12 over $ \mathbb{F}_{2^{10}} $. Thanks to the nice work \cite{Kaspers-Zhou-2020}, we   can  obtain these representatives. In fact, let $ \gamma $ be a primitive element in $ \mathbb{F}^{\ast}_{2^5} $, according to \cite[Theorem 4.5]{Kaspers-Zhou-2020}, there are exactly 6 of CCZ-inequivalent Taniguchi APN functions from F12: $ i=1 $, take $ \alpha=1 $, $ \beta=1,~\gamma^7,~\gamma^{11} $; $ i=2 $, take $ \alpha=1 $, $ \beta=1,~\gamma^3, ~\gamma^{15} $. The notations $ i,~\alpha,~\beta $ used here are the same as the ones used in family F12 of Table \ref{table2}.
\end{rmk}

\begin{rmk} 
Let $ n=2m $ with $ m $ odd, and $ {\rm gcd}(m,3)=1 $. Let $ q=2^m $. Let $ z $ be a primitive element in $ \mathbb{F}^{\ast}_{2^n} $, and $ \omega=z^{\frac{2^n-1}{3}} $. Then $\omega$ is a primitive element in $ \mathbb{F}_{2^2} $. Let $ s=m-2 $ or $ (m-2)^{-1} {\rm ~mod}~n $. Then $   g_{s}(x)=x^3+\omega x^{2^s+1}+\omega^2 x^{3q}+x^{(2^s+1)q} $  is an APN function (\cite{Budaghyan-Helleseth-Kaleyski-2020}). It can be seen that $ g_{s}(x) $ can be covered by our theorem. In fact, noting that $ \omega^{2^s}=\omega^2 $ for any odd $ s $, $ g_{s}(x)=\omega{\rm Tr}^{n}_{m}(\omega^2 x^3)+\omega^2{\rm Tr}^{n}_{m}(\omega^2 x^{2^s+1})=a{\rm Tr}^{n}_{m}(b x^3)+a^q{\rm Tr}^{n}_{m}(c x^{2^s+1})$, where $ a=\omega, b=c=\omega^2$. It is clear that $ a+a^q=1\neq 0 $, and $ b=\omega$ is a non-cube since $ {\rm gcd}(m,3)=1 $, and $ \frac{c^4}{b}=1=\frac{c^{2^t-1}}{b^{2^{2t}}} $, where $ t=(m-2)^{-1} {\rm ~mod}~n$. Then by $ i) $, $ ii) $ of the above theorem, we have that $ g_{s}(x) $ is APN over $ \mathbb{F}_{2^n} $, for $ s=m-2 $, and $ (m-2)^{-1} {\rm ~mod}~n $, respectively. 
\end{rmk} 	

\begin{rmk}\label{f(m-2)}	
	Let $ n=2m $ with $ m $ odd. Let us investigate the APN property of $ f_{m-2}(x) $ further. A pair ($ b,c $) is said to satisfy property $\mathbf{P}_{m-2}$, if $ b $ is a cube in $ \mathbb{F}^{\ast}_{2^n} $, and $ c\in \mathbb{F}^{\ast}_{2^n} $ such that the following assertion holds:
	\begin{center}
		For any $ x\in \mathbb{F}_{2^n}$ with $x\neq 0,1 $,  
		$ x+x^2 $ is a non-cube in $ \mathbb{F}_{2^n}$, if $ \frac{c^4}{b}\cdot  \frac{x^q+x^4}{x+x^2} \in \mathbb{F}_{2^m}$.
	\end{center}
	\noindent Then $ f_{m-2}(x) $ is APN over $ \mathbb{F}_{2^n}$ for these $ b $, $ c $. In fact, this assertion can be seen from the proof of $ i) $ in the above theorem. With the help of computer, we find that when $ m=5$, $7 $, there exist a lot of pairs ($ b,c $) satisfying  $\mathbf{P}_{m-2}$. More precisely, let  $m=5$ or $ 7 $,  $ z $ be a primitive element in $ \mathbb{F}^{\ast}_{2^{2m}} $,
	$ j=\frac{(2^m+1)}{3} $, and $ U=\{(z^j)^i~|~{\rm gcd}(3,i)=1,~1\leq i\leq 2^n-1\}$. Then any pair ($ b,c $) with $ b\neq 0 $ a cube, and $\frac{c^4}{b}\in U$ satisfies $\mathbf{P}_{m-2}$. However, when $ m=9$, $11$, there does not exist such ($ b,c $). We therefore propose the following:\\
	
	{\bf Open Problem 1.}~~Does there  exist infinite odd  integer $ m \geq 1 $ such that  $\mathbf{P}_{m-2}$ holds?
	
\end{rmk}

\begin{rmk}\label{rmk-f_{m}} Let $ n=2m $ with $ m $ odd, and $ q=2^m $. Let us revisit the function $ f_{m}(x)=a{\rm Tr}^{n}_{m}(bx^{3})+a^q{\rm Tr}^{n}_{m}(cx^{2^m+1}) $ investigated in $ v) $. Replacing $ bx^3 $ by $ bx^{2^i+1} $, we let $ f(x)=a{\rm Tr}^{n}_{m}(bx^{2^i+1})+a^q{\rm Tr}^{n}_{m}(cx^{2^m+1})$, where $ i $ is an odd positive integer with $ {\rm gcd}(i,m)=1 $. With similar arguments, by  $ 3~|~2^i+1 $ and $ {\rm gcd}(i,m)=1 $, we can obtain that $ f(x) $ is APN, if $ b $ is a non-cube in $ \mathbb{F}_{2^n}$, and $ c\notin \mathbb{F}_{2^m} $. Note that $ \frac{1}{a}f(x)=dx^{2^m+1}+{\rm Tr}^{n}_{m}(bx^{2^i+1})$, where  $ d=a^{q-1}(c+c^q) $ can be chosen as any element in $  \mathbb{F}_{2^n}\backslash\mathbb{F}_{2^m} $, since $ a,~c\notin \mathbb{F}_{q} $, we have that $ f(x) $ in fact are exactly the functions in family F10 up to EA-equivalence. This observation suggests that it is worthy to finding APN functions with the following form: 
\begin{eqnarray}\label{f_{i,s}}
 f_{i,s}(x)=a{\rm Tr}^{n}_{m}(bx^{2^i+1})+a^q {\rm Tr}^{n}_{m}(cx^{2^s+1}),~\text{where~} a\in \mathbb{F}_{2^n}~such~that~ a+a^q\neq 0,~n=2m~\text{is a positive integer}.
\end{eqnarray} 	

\end{rmk}

\begin{rmk} It is noted  that there does not exist  elements satisfying the conditions in $ iv) $. However, we decide to preserve this item, because we feel that the technique used in the proof  may  provide some insights for the constructions of APN functions. 
\end{rmk}

\subsection{F, G are both quadratic binomials}

Let us consider more general case. Let $ n=2m $ with $ m $ a positive integer. Let
\begin{eqnarray}\label{h_{i,s,b,c,d,e}}
h_{i,s, b,c,g,e}(x)=a{\rm Tr}^n_{m}(bx^{2^i+1}+cx^{2^{i+m}+1})+ a^q{\rm Tr}^n_{m}(gx^{2^{s}+1}+ex^{2^{s+m}+1}) ,
\end{eqnarray} 	
where $ a\in \mathbb{F}_{2^n} $ such that $ a+a^q\neq 0 $, $ b,c,g,e\in \mathbb{F}_{2^n} $. 

In this subsection, we want to find APN functions of the form (\ref{h_{i,s,b,c,d,e}}). We remark first that the APN  polynomials considered in family F3 can be covered by $ h_{i,s, b,c,g,e}(x) $. In fact, let $ i=m$, $ b\notin \mathbb{F}_{2^m} $, $c=0 $, $ g=1 $, then (\ref{h_{i,s,b,c,d,e}}) becomes  $ a^{q-1}(b+b^q)x^{q+1}+x^{2^s+1}+x^{(2^s+1)q}+ex^{2^sq+1}+e^qx^{2^s+q} $, which are exactly the functions in F3, since $a^{q-1}(b+b^q)$ can be choosen as any elements in  $\mathbb{F}_{2^n}\backslash \mathbb{F}_{2^m}$. 

We can find two infinite families of APN functions  with the above form (\ref{h_{i,s,b,c,d,e}}), and  computationally prove that they are CCZ-inequivalent to any APN power functions over $ \mathbb{F}_{2^{10}} $, and we can find a new sporadic  instance  of APN functions over $ \mathbb{F}_{2^{10}} $.

\begin{thm}\cite{Williams}\label{Williams} Let $n=2m$, and $a\in \mathbb{F}^{\ast}_{2^n}$. Let $t_{1}$ be one solution in $\mathbb{F}_{2^n}$ of $t^2+at+1=0$ (if $ {\rm Tr}^{n}_{1}\Big(\frac{1}{a^2}\Big)=0 $). Let $f(x)=x^3+x+a$, then 

$\bullet$ $ f $ has no zeros in  $\mathbb{F}_{2^n}$ if and only if ${\rm Tr}^{n}_{1}\Big(\frac{1}{a^2}\Big)=0$, and $t_{1}$ is not a cube  in $\mathbb{F}_{2^n}$.

$\bullet$ $ f $ has three zeros in  $\mathbb{F}_{2^n}$ if and only if ${\rm Tr}^{n}_{1}\Big(\frac{1}{a^2}\Big)=0$, and $t_{1}$ is a cube  in $\mathbb{F}_{2^n}$.
\end{thm}

We need the following theorem, which will be used for generating APN functions (see Corollary \ref{corollary}). Let $ n=2m $ with $ m $ being an odd positive integer, and $ q=2^m $. Let $ x\in \mathbb{F}_{2^n}$ with $ x\neq 0, 1 $. Then fix the following notations for this given element $x$.
\begin{eqnarray*}
	&&r:=x^{q+1};~h:=x+x^q;~c:=x+x^2;\\
	&&D:=A(A^{q+1}+B^{q+1});~ H:=A^2(A^qB^3+AB^{3q}+B^{2+2q}),
\end{eqnarray*} 
where $ A, B  $ are some elements determined by $ x $. 
By a routine work, we have that
\begin{eqnarray*} 
	h+h^2=c+c^q.
\end{eqnarray*} 

The following result can not only give rise to APN functions of the form (\ref{h_{i,s,b,c,d,e}}) but   can also  yield Budaghyan-Carlet APN hexanomials  (family F3), and hence it has its own
importance and we state it as a theorem. The proof can be seen in the appendix. 

\begin{thm}\label{vital} Let $ n=2m $ with $ m $ being an odd positive integer. Let $x$ be any given element in $\mathbb{F}_{2^n} \backslash \{0,1\}$. Use the notations given as above. Let 
	\begin{eqnarray}\label{key-1}
f(y)=Ay^3+By^2+B^qy+A^q=0.
	\end{eqnarray} 
Then equation (\ref{key-1}) has no solutions in  $\mathbb{F}_{2^n}$, if A, B, c satisfy 

1) 	$ A=c^{2-2q}(h+c+c^2) $, $ B=c+c^2 $, and $ c=x+x^2$ is a non-cube in $ \mathbb{F}_{2^n}$; or 

2)  $ A=\frac{h+c+c^2}{c^q} $, $ B=1+c$, and $ c=x+x^2$ is a non-cube in $ \mathbb{F}_{2^n}$. 
\end{thm}	
\begin{rmk} Let $ n=2m $, and $ q=2^m $. Recall first that the condition needed in family F3 is that 
\begin{eqnarray}\label{F3}
y^{2^i+1}+dy^{2^i}+d^qy+1=0
\end{eqnarray} 	
has no solutions in $ U=\{x\in \mathbb{F}_{2^n}~|~x^{q+1}=1\}$.  	
Here $ i $ is a positive integer with $ {\rm gcd}(i,m)=1 $. When $ i=1 $, this condition is exactly that $y^{3}+dy^{2}+d^qy+1=0$	
has no solutions in $ U $.

With the same notations as   in Theorem \ref{vital}. Let $ A $ be the elements given in 1) or 2). Let  $\Gamma=\{A\in \mathbb{F}^{\ast}_{2^m}~|~x\in \mathbb{F}_{2^n}\backslash \mathbb{F}_{2^m},~c=x+x^2~\text{not cube}\}$.  Numerical experiments suggest that $ \Gamma $ is always nonempty for any odd $ m $. This can yield Budaghyan-Carlet APN functions in family F3.    In fact, let $A\in \Gamma$, then (\ref{key-1}) becomes 
	\begin{eqnarray*}
		y^3+dy^2+d^qy+1=0,~d=\frac{B}{A}. 
	\end{eqnarray*} 
	According to Theorem \ref{vital}, the above equation has no solutions in  $ \mathbb{F}_{2^n} $.  Therefore, this theorem can be used to yield APN functions in family F3. It is noted  that the existence of the coefficients $ d $ such that  the equation (\ref{F3}) has no solutions in $ U $ (or $ \mathbb{F}_{2^n} $) for a given positive integer $ i $ had also been studied in \cite{Bluher-2013,Bracken-Tan-Tan-2014}.We expect that $\Gamma$ does indeed empty for any odd positive integer $ m $, and hence propose the following: 
	
	{\bf Open problem 2.} Let $ n=2m $ with $ m $ odd. Show that $ \Gamma $ is always nonempty. 
	
\noindent It is  also interesting and important to  consider the following question.

	{\bf Open problem 3.} Let $ n=2m $ with $ m $ a positive integer, $ q=2^m $. Let $ i $ be a positive with $ {\rm gcd}(m,i)=1 $. Find more exponents $ i $, and elements $ A, B $ such that the following equation has no solutions in $ \mathbb{F}_{2^n} $.
	\begin{eqnarray*}
		Ay^{2^i+1}+By^{2^i}+B^qy+A^q=0. 
	\end{eqnarray*} 	
\end{rmk}

In the following, we investigate the APN property of the functions with the form (\ref{h_{i,s,b,c,d,e}}) by letting $ i=1, c=0$. We does indeed find two infinite families of APN functions. But, astonishingly enough, the function obtained happened to be  CCZ-equivalent to some functions in family F12 with a completely different from that of Taniguchi.

\begin{corollary}\label{corollary} Let $n=2m$ be a positive integer with  $ m $ odd, and $ q=2^m $. Let $ h_{s}(x)=a{\rm Tr}^n_{m}(bx^3)+ a^q{\rm Tr}^n_{m}(gx^{2^{s}+1}+ex^{2^{s+m}+1})$ with $ a\notin \mathbb{F}_{q} $, $ bge\neq 0 $. Then $ h_{s}(x) $ is APN over $ \mathbb{F}_{2^n} $, if $ s, b, g, e $ satisfy  
	\begin{eqnarray*}
	&1)&~~~s=2,~b{\rm~is~not~a~cube},~g=1,~ e=\frac{1}{b^{2q-2}}; {~~\rm or} \\
	&2)&~~~s=2,~b{\rm~is~not~a~cube},~g=e=b.
	\end{eqnarray*}

\end{corollary}
\begin{proof} 1) $ s=2 $,~$ b ${\rm~is~not~a~cube},~$ g=1 $,~ $ e=\frac{1}{b^{2q-2}} $. 
	
	Let $ F(x)=bx^3 $, $ G(x)=x^{2^s+1}+ex^{2^{s+m}+1} $. Then we have 
	\begin{eqnarray*}
		\Delta_{d,F}=d^3b(x+x^2), ~ \Delta_{d,G}=d^{2^s+1}(x+x^{2^s})+d^{2^{s+m}+1}e(x+x^{2^{s+m}}).
	\end{eqnarray*} 
	According to Lemma \ref{fundamental-lemma}, we have that $ h_{s}(x) $ is APN  if the following system
	\begin{eqnarray*}
		\begin{cases}
			d^3b(x+x^2)=\alpha&\\  
			d^{2^s+1}(x+x^{2^s})+d^{2^{s+m}+1}e(x+x^{2^{s+m}})=\beta &
		\end{cases}
	\end{eqnarray*} 
	only has $ x=0, 1  $ as its solutions for any $ d\neq 0 $, where $ \alpha$,~$\beta \in \mathbb{F}_{2^m} .$ Assume, to the contrary, that there exists some $ d\neq 0 $, $ x\neq 0,1 $ such that the above system holds. Now let $ s=2 $, $ b $ is a non-cube, $ e=\frac{1}{b^{2q-2}} $. Then $ \alpha\neq 0 $, $ b=\frac{\alpha}{d^3(x+x^2)} $, $ e=b^{-(2q-2)}=d^{6q-6}(x+x^2)^{2q-2} $ (note that $ \alpha^{2q-2}=1 $). Substituting it into the second equation of the above system, we have 
	\begin{eqnarray*}
		d^5(x+x^4)+d^{10q-5}(x+x^2)^{2q-2}(x+x^{4q})=\beta,
	\end{eqnarray*} 
	which is equivalent to 
	\begin{eqnarray}\label{h-1}
	d^5(x+x^4)+d^{10q-5}(x+x^2)^{2q-2}(x+x^{4q})+\Big(d^5(x+x^4)+d^{10q-5}(x+x^2)^{2q-2}(x+x^{4q})\Big)^q=0.
	\end{eqnarray} 
	Let $u=d^5$. Then the above equation becomes
	\begin{eqnarray}\label{h-2}
	u(x+x^4)+u^{2q-1}(x+x^2)^{2q-2}(x+x^{4q})+\Big(u(x+x^4)+u^{2q-1}(x+x^2)^{2q-2}(x+x^{4q})\Big)^q=0.
	\end{eqnarray} 
 Note that any nonzero element $ u $ of $ \mathbb{F}_{2^n} $ has a unique polar decomposition of the form $ u=vk $, where $ v^{q+1}=1 $, and $ k^{q-1}=1 $. Substituting $ u=vk $ into (\ref{h-2}), then (\ref{h-2}) can be reduced as
\begin{eqnarray*}
	v(x+x^4)+v^{2q-1}(x+x^2)^{2q-2}(x+x^{4q})+\Big(v(x+x^4)+v^{2q-1}(x+x^2)^{2q-2}(x+x^{4q})\Big)^q=0.
\end{eqnarray*} 
Multiplying both sides by  $ v^3 $ of the above equation, by the fact that $ v^{q}=v^{-1} $, we have
\begin{eqnarray*}
	Ay^3+By^2+B^qy+A^q=0,
\end{eqnarray*} 
where $ y=v^2\in \mathbb{F}_{2^n}$, and  $ A $, $ B $ are given in 1) of Theorem \ref{vital}. Now, according to 1) of Theorem \ref{vital}, we obtian that the element $ x+x^2 $ is a cube, and hence $ b $ is a cube from  the first equation $ d^3b(x+x^2)=\alpha $ of the system, since $ \alpha \in \mathbb{F}^{\ast}_{2^m} $ is a cube. This derives  a contradiction to the assumption that $ b $ is a non-cube.

2) $ s=2 $,~$ b ${\rm~is~not~a~cube},~$ g=e=b $.

Let $F(x)=bx^3$ and $G(x)=bx^5+bx^{4q+1}$. We have
\begin{align*}
\Delta_{d,F}(x)=d^3b(x+x^2)\hspace{0.2cm} {\rm and}\hspace{0.2cm} \Delta_{d,G}(x)=d^5b(x+x^4)+d^{4q+1}b(x+x^{4q}).
\end{align*}
By Lemma 2.1, $h_{s}(x)$ is APN if and only if the following system
\begin{align*}
\begin{cases}
d^3b(x+x^2)=\alpha\\
d^5b(x+x^4)+d^{4q+1}b(x+x^{4q})=\beta
\end{cases}
\end{align*}
only has trivial solutions $x\in\mathbb{F}_2$ for any $d\in\mathbb{F}_{2^n}^*$ and $\alpha, \beta\in\mathbb{F}_{2^m}$. Assume now that there exist some $d\in\mathbb{F}_{2^n}^*$,  $ \alpha \in\mathbb{F}_{2^m}  $, $\beta\in\mathbb{F}_{2^m}$ such that the system has non-trivial solutions $x\in\mathbb{F}_{2^n}\backslash\mathbb{F}_2$. Then $ \alpha\neq 0 $. By the first equation, we have $b=\frac{\alpha}{d^3(x+x^2)}$. Substituting this relation into the second equation, we have
\begin{align*}
\frac{d^2(x+x^4)}{x+x^2}+\frac{d^{4q-2}(x+x^{4q})}{x+x^2}=\frac{\beta}{\alpha},
\end{align*}
which implies that
\begin{align*}
\frac{d^2(x+x^4)}{x+x^2}+\frac{d^{4q-2}(x+x^{4q})}{x+x^2}+\bigg(\frac{d^2(x+x^4)}{x+x^2}+\frac{d^{4q-2}(x+x^{4q})}{x+x^2}\bigg)^q=0,
\end{align*}
since $\alpha,~\beta\in\mathbb{F}_{2^m}$. Let $\mu=d^2$. We have
\begin{align}\label{eq1}
\frac{\mu(x+x^4)}{x+x^2}+\frac{\mu^{2q-1}(x+x^{4q})}{x+x^2}+\bigg(\frac{\mu(x+x^4)}{x+x^2}+\frac{\mu^{2q-1}(x+x^{4q})}{x+x^2}\bigg)^q=0.
\end{align}
To complete the proof, it suffices to show that $x+x^2$ is a cube of $\mathbb{F}_{2^n}$, which will derive that $ b $ is a cube from the first equation of the above system and this will yield a contradiction to the assumption that $ b $ is a  non-cube. Let  $\mu=\nu k$, where $\nu^{q+1}=1$ and $k\in\mathbb{F}_{2^m}^*$, and substitute $\mu=\nu k$ into \eqref{eq1}, we have
\begin{align*}
\frac{\nu(x+x^4)}{x+x^2}+\frac{\nu^{2q-1}(x+x^{4q})}{x+x^2}+\bigg(\frac{\nu(x+x^4)}{x+x^2}+\frac{\nu^{2q-1}(x+x^{4q})}{x+x^2}\bigg)^q=0.
\end{align*}
Multiplying both sides of the above equation by $\nu^3$, we have
\begin{align*}
Ay^3+By^2+B^qy+A^q=0,
\end{align*}
where $y=\nu^2$, $A=\Big(\frac{x+x^{4q}}{x+x^2}\Big)^q$ and $B=\frac{x+x^4}{x+x^2}=1+x+x^2$. According to 2) of Theorem \ref{vital}, $ x+x^2 $ is a cube in $ \mathbb{F}_{2^n} $, otherwise, the above equation has no solutions in $ \mathbb{F}_{2^n} $. \end{proof}

{\bf Example 1}. Besides the two infinite classes of APN functions presented in Corollary \ref{corollary}, we can also find a new instance of APN  functions over $ \mathbb{F}_{2^{10}} $ CCZ-inequivalent to any other known APN functions. Let $ z $ be a primitive element in $ \mathbb{F}^{\ast}_{2^{10}} $. Then 
\begin{eqnarray*}\label{h_{1,2, b,0,d,e}-instance}
h_{s}(x)=a{\rm Tr}^n_{m}(bx^{3})+ a^q{\rm Tr}^n_{m}(gx^{5}+ex^{4q+1})
\end{eqnarray*} 
is an APN function  over $ \mathbb{F}_{2^{10}} $, where $ b=1 $, $ g=z $, $ e=z^{369}$.

\begin{table}[h]
	\centering
	\caption{All Known CCZ-inequivalent APN  functions over $ \mathbb{F}_{2^{10}} $, $ q=2^5 $}
	\label{table3} 
	\centering
	\begin{tabular}{|m{150pt}|m{117pt}|m{55pt}|}
		\hline
		Function  & Conditions  & Family \\  
		\hline
		$  x^{2^i+1} $ & $ i=1, 3 $ &  Gold \\
		\hline
		$ x^{57} $ & $ -$  & Kasami \\
		\hline
		$ x^{339} $ & $ -$  &  Dobbertin \\
		\hline
		$ x^{6}+x^{33}+\alpha^{31}x^{192} $ & $ \alpha $ primitive in $ \mathbb{F}^{\ast}_{2^{10}} $  & F3  \\
		\hline
		$ x^{33}+x^{72}+\alpha^{31}x^{258} $ &$ \alpha $ primitive in $ \mathbb{F}^{\ast}_{2^{10}} $ & F3 \\
		\hline
		$ x^3+{\rm Tr}^{10}_{1}(x^9) $ & $ -$ 
		& F4 \\
		\hline
		$ x^3+\alpha^{-1} {\rm Tr}^{10}_{1}(\alpha ^3x^9) $	& $ \alpha $ primitive in $ \mathbb{F}^{\ast}_{2^{10}} $ & F4 \\
		\hline
		\tabincell{l}{$ u(u^qx+ux^q)(x+x^q)+$\\ $(u^qx+ux^q)^{2^{2i}+2^{3i}}+$\\ $\alpha(u^qx+ux^q)^{2^{2i}}(x+x^q)^{2^i}+$\\ $\beta(x+x^q)^{2^{i}+1}$ }  & \tabincell{l}{ $ u $ primitive in $\mathbb{F}^{\ast}_{2^{10}}$, \\ $z$ primitive in $ \mathbb{F}^{\ast}_{2^5} $, \\ $i=1$, 
			$\alpha=1 $, $ \beta=1 , z^7, z^{11}$;\\
			$i=2$, 	$\alpha=1 $, $ \beta=1 , z^3, z^{15}$}   & F12 \\
		\hline
		$ B(x)=x^3+\alpha^{341}x^{36} $	& $-$  & sporadic, see \cite{Edel-Kyureghyan-Pott-2006}  \\
		\hline
		\tabincell{l}{$ x^3+\omega x^{2^s+1}+$$\omega^2x^{3q}+x^{(2^s+1)q} $}	& 	\tabincell{l}{$s=3,7,$  $\omega$ primitive in $ \mathbb{F}^{\ast}_{2^2} $\\} & F14 \\
		\hline			
		\tabincell{l}{$ \alpha{\rm Tr}^{n}_{m}(\alpha x^3)+\alpha^q {\rm Tr}^{n}_{m}(\alpha^{3}x^9) $}&	\tabincell{l}{ $ \alpha $ primitive in $ \mathbb{F}^{\ast}_{2^{10}}$ }   & F15 \\
		\hline	
		\tabincell{l}{$ \alpha{\rm Tr}^{n}_{m}(x^3)+\alpha^q {\rm Tr}^{n}_{m}(\alpha^{11}x^9) $}&	\tabincell{l}{ $ \alpha $ primitive in $ \mathbb{F}^{\ast}_{2^{10}}$ }   &\tabincell{l}{ sporadic, see\\ Remark \ref{f(m-2)}}\\
		\hline		
		\tabincell{l}{$ \alpha{\rm Tr}^{n}_{m}( x^3)+\alpha^q {\rm Tr}^{n}_{m}(\alpha x^5+\alpha^{369}x^{4q+1}) $}&	\tabincell{l}{ $ \alpha $ primitive in $ \mathbb{F}^{\ast}_{2^{10}}$ }   & \tabincell{l}{ sporadic, see\\ Example 1} \\
		\hline	
	\end{tabular}
\end{table}

\section{Conclusions}

Let $ n=2m $, and $ q=2^m $. We studied a class of quadratic functions with the form
$ f(x)=a{\rm Tr}^{n}_{m}(F(x))+a^q{\rm Tr}^{n}_{m}(G(x))$, where $ F $, $ G $ are  quadratic functions. We found a new infinite family of APN quadrinomials over $ \mathbb{F}_{2^n} $, $ a\in \mathbb{F}_{2^n} $, $ n=2m $ with $ m $ odd as follows.
\begin{eqnarray*}
f_{1}(x)=a{\rm Tr}^{n}_{m}(bx^3)+a^q{\rm Tr}^{n}_{m}(b^3x^9), ~b \text{~not~a~cube},~a\notin \mathbb{F}_{q}.
\end{eqnarray*} 
We generalized the two infinite families of APN functions obtained in \cite{Budaghyan-Helleseth-Kaleyski-2020} to a broader condition on $ m $, that is, the assumption that $ {\rm gcd}(3,m)=1 $ needed in \cite{Budaghyan-Helleseth-Kaleyski-2020} can be removed, up to CCZ-equivalence. We also found two infinite families of APN functions over $ \mathbb{F}_{2^{2m}} $ for odd $ m $, which turned out to be in family F12, that is, the the Taniguchi APN functions when $ m=5 $, as follows.
\begin{eqnarray*}
	f_{2}(x)=a{\rm Tr}^{n}_{m}(bx^3)+a^q{\rm Tr}^{n}_{m}(x^5+\frac{1}{b^{2q-2}}x^{4q+1}), ~b \text{~not~a~cube},~a \in \mathbb{F}_{2^n} \backslash \mathbb{F}_{2^m},
\end{eqnarray*} 
and 
\begin{eqnarray*}
	f_{3}(x)=a{\rm Tr}^{n}_{m}(bx^3)+a^q{\rm Tr}^{n}_{m}(bx^5+bx^{4q+1}), ~b \text{~not~a~cube},~a \in \mathbb{F}_{2^n} \backslash \mathbb{F}_{2^m}.
\end{eqnarray*} 
Code isomorphism tests showed that $ f_{2} $ and $ f_{3}$ are CCZ-inequivalent to each other over $ \mathbb{F}_{2^{10}}$. We found two new instances of APN functions over $ \mathbb{F}_{2^{10}} $. We also proposed three open problems, and we cordially invite the readers to attack these open problems.

\section{Appendix}

\subsection{Proof of 1) in Theorem \ref{vital}}

\begin{proof} It can be checked that $ A^{q+1}+B^{q+1}=(x+x^q)^5=h^5 $ in this case. In the following, we assume that $ c $ is a non-cube in $ \mathbb{F}_{2^n} $. Note that $ A\neq 0 $. In fact, if $ A=0 $, then $ h+c+c^2=x^4+x^q =0$, which implies that $ x\in \mathbb{F}_{2^n}\cap \mathbb{F}_{2^{m-2}}=\mathbb{F}_{2}$, since $ m $ is odd, and $ {\rm gcd}(n,m-2)=1 $, a contradiction to the assumption that $ x\neq 0, 1 $. Let $ y:=y+\frac{B}{A} $. Then equation $ (\ref{key-1}) $ becomes 
	\begin{eqnarray*}
		y^3+\frac{AB^q+B^2}{A^2}y+\frac{A^{q+1}+B^{q+1}}{A^2}=0.
	\end{eqnarray*} 
	
	Let $ y=Ez $, where $ E $ satisfies that $ E^2=\frac{AB^q+B^2}{A^2} $. Note that $ E\neq 0 $. In fact, this would imply that $ AB^q=B^2 $, and hence $ A^qB=B^{2q} $, $ A^{q+1}B^{q+1}=B^{2(q+1)}$. However, by the fact that $ B\neq 0 $~(if $ B=0 $, then $ c=0, 1 $, a contradiction to the assumption that $ c $ is a non-cube), we have $ A^{q+1}+B^{q+1}=0 $, which implies that $ (x+x^q)^5=0 $, i.e., $ x\in \mathbb{F}_{q} $, and then $ c=x+x^2\in \mathbb{F}_{2^m} $ is a cube in $ \mathbb{F}_{2^n} $, since every element in $ \mathbb{F}_{2^m} $ is a cube by the fact that ${\rm gcd}(3,2^m-1)=1 $~(since $m$ is odd), a contradiction.

	Then the above equation becomes
	\begin{eqnarray}\label{key-2}
	z^3+z+a=0, 
	\end{eqnarray} 
	where $ a\neq 0 $ satisfies that 
	\begin{eqnarray*}
		a=\frac{A^{q+1}+B^{q+1}}{A^2E^3}.
	\end{eqnarray*} 
	It can be checked that 
	\begin{eqnarray}\label{1/a^2-1}
	\frac{1}{a^2}=\frac{(AB^q+B^2)^3}{A^2(A^{q+1}+B^{q+1})^2}.
	\end{eqnarray} 
	It is clear that equation (\ref{key-1}) has no solutions in $ \mathbb{F}_{2^n} $ if and only if (\ref{key-2})  has no solutions. To complete the proof, according to Theorem \ref{Williams}, we have to show that $ {\rm Tr}^{n}_{1}\Big(\frac{1}{a^2}\Big)=0 $, and $ t_{1} $ is a non-cube in   $ \mathbb{F}_{2^n} $, where $ t_{1} $ is one solution in $ \mathbb{F}_{2^n} $ of $ t^2+at+1=0 $.
	
	{\bf Claim 1.} $ {\rm Tr}^{n}_{1}\Big(\frac{1}{a^2}\Big)=0 $.
	
	In fact, we have  
	\begin{eqnarray}\label{1/a^2-2}
	\frac{1}{a^2}=\frac{(AB^q+B^2)^3}{A^2(A^{q+1}+B^{q+1})^2}=\frac{B^3+M}{A(A^{q+1}+B^{q+1})}+\Bigg(\frac{B^3+M}{A(A^{q+1}+B^{q+1})}\Bigg)^2,
	\end{eqnarray} 
	where $ M $  is one solution of the following equation
	\begin{eqnarray}\label{M}
	M^2+DM+H=0.
	\end{eqnarray} 
	Recall the notations that  $ D=A(A^{q+1}+B^{q+1}) $, $ H=A^2(A^qB^3+AB^{3q}+B^{2+2q}) $, $ A=c^{2-2q}(h+c+c^2) $, $ B=c+c^2 $, $ c=x+x^2 $.
	We need only to show that the above equation in $M$ has solutions in $ \mathbb{F}_{2^n} $, i.e., $ {\rm Tr}^{n}_{1}\Big({\frac{H}{D^2}}\Big)=0 $. This can be seen from the following fact. 
	\begin{eqnarray*}
		\frac{H}{D^2}=\frac{A^qB^3+AB^{3q}+B^{2+2q}}{(A^{q+1}+B^{q+1})^2}
	\end{eqnarray*} 
	is an element in $ \mathbb{F}_{2^m} $, since $ A^qB^3+AB^{3q}={\rm Tr}^{n}_{m}(A^qB^3)$,~$  A^{q+1} $, $ B^{q+1}\in \mathbb{F}_{2^m} $.
	
	Next, we need to find one solution $ t_{1} $ in $ \mathbb{F}_{2^n} $ of $ t^2+at+1=0 $, and show that $ t_{1} $ is a non-cube. It is clear that $ t_{1}$ can be represented as $av$, where $ v=\frac{B^3+M}{A(A^{q+1}+B^{q+1})}$, since  $ \frac{1}{a^2}=v+v^2 $ according to  (\ref{1/a^2-2}). Note that $ t_{1}=av $ satisfies that 
	\begin{eqnarray*}
		t^2_{1}=a^2v^2=\frac{(B^3+M)^2}{(AB^q+B^2)^3}.
	\end{eqnarray*} 
	Therefore, to show $ t_{1} $ is a non-cube in $ \mathbb{F}_{2^n} $, we have to show that $ B^3+M $ is a non-cube.

	{\bf Claim 2.} $ B^3+M $ is a non-cube in $ \mathbb{F}_{2^n} $.
	
	Our strategy is to find the explicit expression of $M$, and then show that $ B^3+M $ is a non-cube. To this end, we have to revisit equation ($ \ref{M} $), and explore more information on the element $ \frac{H}{D^2}$~(it is in $ \mathbb{F}_{2^m} $). Very fortunately, we  find  that $ {\rm Tr}^{m}_{1}\Big(\frac{H}{D^2}\Big)=0. $ In fact, recall the notations that $ h=x+x^q $, and $ r=x^{q+1} $, we find~(with computer assistance)~that~(a surprise)
	\begin{eqnarray}\label{H/D^2}
	\frac{H}{D^2}=u+u^2,
	\end{eqnarray} 
	where 
	\begin{eqnarray*}
		u=\frac{h^2(r+r^2)+r+r^4+hr^2}{h^5}.
	\end{eqnarray*} 
	Then $ M $ can be chosen as $ Du $~(this is because it suffices   to find  one solution of $ M^2+DM+H=0 $). We find that 
	\begin{eqnarray*}
		M=Du=A(A^{q+1}+B^{q+1})u&=&c^{2-2q}(h+c+c^2)h^5\cdot \frac{h^2(r+r^2)+r+r^4+hr^2}{h^5}\\
		&=&\frac{c^2(h+c+c^2)({h^2(r+r^2)+r+r^4+hr^2)}}{c^{2q}}.
	\end{eqnarray*} 
	Then, recall the notation that $ B=c+c^2 $, we can obtain the expression of $ B^3+M $ as follows. 
	\begin{eqnarray}\label{B^3+M}
	\begin{aligned}
	B^3+M&=&\frac{h(c^{2q+4}+c^{q+5}+c^{q+4})+c^{2q+4}+c^{q+5}}{c^{2q}}\\
	&=&\frac{c^2\Big(h(c^{2q+4}+c^{q+5}+c^{q+4})+c^{2q+4}+c^{q+5}\Big)}{c^{2+2q}}.
	\end{aligned}
	\end{eqnarray} 
	The above expression can be deduced from
	\begin{eqnarray*}
		h+h^2=c+c^q,~c^{q+1}=r+r^2+hr.
	\end{eqnarray*} 
	Note that $ h,~c^{2+2q}\in \mathbb{F}^{\ast}_{2^m} $ is a cube, it suffices to show that 
	\begin{eqnarray*}
		hc^2\Big(h(c^{2q+4}+c^{q+5}+c^{q+4})+c^{2q+4}+c^{q+5}\Big)
	\end{eqnarray*} 
	is a non-cube. By the fact that $ h+h^2=c+c^q $, we have 
	\begin{eqnarray*}
		hc^2\Big(h(c^{2q+4}+c^{q+5}+c^{q+4})+c^{2q+4}+c^{q+5}\Big)=c^5c^{q+1}((c+c^q)^2+h^2).
	\end{eqnarray*} 
	Since $ c^{q+1}$, $c+c^q $, $ h \in \mathbb{F}^{\ast}_{2^m} $ are all cubes in $ \mathbb{F}_{2^n} $, we have that the above element is a non-cube,  when $ c $ is a non-cube.
\end{proof}	

\subsection{Proof of 2) in Theorem \ref{vital}} 
\begin{proof}	
The proof is similar to that of 1) in Theorem \ref{vital}.
Recall the following notations:
$ r=x^{q+1}; h= x+x^{q};c=x+x^2; $$ A= \frac{h+c+c^2}{c^q};
B=1+c,$ from which we can obtain that $h+h^2=c+c^q$ and $A^{q+1}+B^{q+1}=\frac{(x+x^q)^5}{(x+x^2)^{q+1}}=\frac{h^5}{c^{q+1}}$.
Note that $A\neq 0$, otherwise, we have $x+x^{4q}=0$ that means that $x\in\mathbb{F}_{2^n}\cap \mathbb{F}_{4q}=\mathbb{F}_2$, since $\gcd(m+2,n)=1$. Then setting $y:=y+\frac{B}{A}$, this can transform \eqref{key-1} into
\begin{align}\label{eq2}
y^3+\frac{AB^q+B^2}{A^2}y+\frac{A^{q+1}+B^{q+1}}{A^2}=0.
\end{align}
Observe that $B\neq0$ (otherwise $c=1$ is a cube) and $AB^q+B^2\neq 0$, otherwise, we have $A^{q+1}+B^{q+1}=0$, that is, $h=0$, which implies that $c\in\mathbb{F}_q$ contracting to the assumption that $c$ is a non-cube, since $\gcd(3,2^m-1)=1$ for any odd $m$. Thus we can transform the equation \eqref{eq2} into
\begin{align}\label{eq3}
z^3+z+a=0
\end{align}
by setting $y=Ez$, where $a,E\in\mathbb{F}_{2^n}^*$ such that
\begin{align*}
E^2=\frac{AB^q+B^2}{A^2} \hspace{0.2cm} {\rm and}\hspace{0.2cm} a^2=\frac{A^2(A^{q+1}+B^{q+1})^2}{(AB^q+B^2)^3}.
\end{align*}
We need now to prove that equation \eqref{eq3} has no solutions in $\mathbb{F}_{2^n}$. According to  Theorem \ref{Williams}, we have to show that ${\rm Tr}_1^n\Big(\frac{1}{a^2}\Big)=0$ and the solutions in $\mathbb{F}_{2^n}$  of equation $t^2+at+1=0$  are not cubes of $\mathbb{F}_{2^n}$.

Firstly, we prove that ${\rm Tr}_1^n\Big(\frac{1}{a^2}\Big)=0$. Note that $\frac{1}{a^2}$ can be written as
\begin{align}\label{eq4}
\frac{1}{a^2}=\frac{B^3+M}{A(A^{q+1}+B^{q+1})}+\bigg(\frac{B^3+M}{A(A^{q+1}+B^{q+1})}\bigg)^2,
\end{align}
where $M$ is a solution of
\begin{align}\label{eq5}
M^2+DM+H=0,
\end{align}
where $D=A(A^{q+1}+B^{q+1})$ and $H=A^2(AB^{3q}+A^{q}B^3+B^{2(q+1)})$. Then we transform the problem into showing  that equation \eqref{eq5} has solutions in $\mathbb{F}_{2^n}$, which is equivalent to  ${\rm Tr}_1^n\Big(\frac{H}{D^2}\Big)=0$. Indeed, it can be seen that
\begin{align*}
\frac{H}{D^2}=\frac{AB^{3q}+A^{q}B^3+B^{2(q+1)}}{(A^{q+1}+B^{q+1})^2}=\frac{{\rm Tr}_m^n(AB^{3q})+B^{2(q+1)}}{(A^{q+1}+B^{q+1})^2},
\end{align*}
which is clearly in $\mathbb{F}_{q}$. Thus, ${\rm Tr}_1^n\Big(\frac{H}{D^2}\Big)=0$.

Then, we show that the solutions of $t^2+at+1=0$ are not cubes in $\mathbb{F}_{2^n}$. Assume that $t_1$ is a solution of $t^2+at+1=0$. Then by \eqref{eq4}, it can be represented by $t_1=a\nu$, where $\nu=\frac{B^3+M}{A(A^{q+1}+B^{q+1})}$, and thus
\begin{align*}
t_1^2=a^2\nu^2=\frac{(B^3+M)^2}{(AB^q+B^2)^3}.
\end{align*}
Therefore, to show $t_1$ is not a cube, it suffices to show  $(B^3+M)^2$ and thus $B^3+M$ is not a cube of $\mathbb{F}_{2^n}$. In the following, we show this fact by giving the explicit expression of $M$ by revisiting \eqref{eq5} again.

By the above discussion, we have obtained that $\frac{H}{D^2}\in\mathbb{F}_q$. We  further  want to show that ${\rm Tr}_1^m\Big(\frac{H}{D^2}\Big)=0$, which is equivalent to showing
\begin{align}\label{eq6}
\frac{H}{D^2}=\mu+\mu^2
\end{align}
for some $\mu\in\mathbb{F}_{2^m}$. Recall that $A=\frac{h+c+c^2}{c^q}$, $B=1+c$ and $A^{q+1}+B^{q+1}=\frac{h^5}{c^{q+1}}$,  we have
\begin{align*}
A^qB^3+AB^{3q}=&\frac{(h+c^q+c^{2q})B^3}{c}+\frac{(h+c+c^{2})B^{3q}}{c^q}\\
=&\frac{c^q(h+c^q+c^{2q})B^3+c(h+c+c^{2})B^{3q}}{c^{q+1}}\\
=&\frac{h(c^qB^3+cB^{3q})+c^qB^3(c^q+c^{2q})+cB^{3q}(c+c^2)}{c^{q+1}}.
\end{align*}
While
\begin{align*}
h(c^qB^3+cB^{3q})=&h\big(c^q(1+c+c^2+c^3)+c(1+c^q+c^{2q}+c^{3q})\big)\\
=&h(c+c^q+c^{q+1}(c+c^q)+c^{q+1}(c+c^q)^2)
\end{align*}
and
\begin{align*}
c^qB^3(c^q+c^{2q})+cB^{3q}(c+c^2)=&c^q(1+c+c^2+c^3)(c^q+c^{2q})+(c^q(1+c+c^2+c^3)(c^q+c^{2q}))^q\\
=&c^{2q}+c^{3q}+c^{2q+1}+c^{3q+1}+c^{2q+2}+c^{3q+2}+c^{2q+3}+c^{3q+3}+\\
&~~~~~~~~(c^{2q}+c^{3q}+c^{2q+1}+c^{3q+1}+c^{2q+2}+c^{3q+2}+c^{2q+3}+c^{3q+3})^q\\
=&(c+c^q)^2+c^3+c^{3q}+c^{q+1}(c+c^q)+c^{q+1}(c+c^q)^2.
\end{align*}
We have
\begin{align*}
c+c^q=x+x^q+(x+x^q)^2,
c^{q+1}=x^{q+1}+x^{2(q+1)}+x^{q+1}(x+x^q),
\end{align*}
from which we can obtain that
\begin{align*}
h(c^qB^3+cB^{3q})=&(x+x^q)^2+(x+x^q)^3+x^{q+1}(x+x^q)^2+x^{q+1}(x+x^q)^3+x^{q+1}(x+x^q)^5\\
&+x^{q+1}(x+x^q)^6+x^{(2q+1)}(x+x^q)^2+x^{(2q+1)}(x+x^q)^5
\end{align*}
and
\begin{align*}
c^qB^3(c^q+c^{2q})+cB^{3q}(c+c^2)=&(x+x^{q})^2+(x+x^{q})^3+(x+x^{q})^5+(x+x^{q})^6+x^{q+1}(x+x^{q})^2\\
&+x^{q+1}(x+x^{q})^3+x^{q+1}(x+x^{q})^4+x^{q+1}(x+x^{q})^5\\
&+x^{2(q+1)}(x+x^{q})^2+x^{2(q+1)}(x+x^{q})^4.
\end{align*}
Thus we have
\begin{align*}
c^{q+1}(A^qB^3+AB^{3q})=&(x+x^{q})^5+(x+x^{q})^6+x^{q+1}(x+x^{q})^4+x^{q+1}(x+x^{q})^6\\
&+x^{2(q+1)}(x+x^{q})^4+x^{2(q+1)}(x+x^{q})^5
\end{align*}
and
\begin{align*}
c^{2(q+1)}(A^qB^3+AB^{3q})=&x^{q+1}(x+x^{q})^5+x^{q+1}(x+x^{q})^7+x^{2(q+1)}(x+x^{q})^4+x^{2(q+1)}(x+x^{q})^7\\
&+x^{4(q+1)}(x+x^{q})^4+x^{4(q+1)}(x+x^{q})^5.
\end{align*}
We further have
\begin{align*}
c^{2(q+1)}B^{2(q+1)}=&c^{2(q+1)}(1+c)^{2(q+2)}\\
=&c^{2(q+1)}+c^{2(q+1)}(c+c^q)^2+c^{4(q+1)}\\
=&x^{2(q+1)}+x^{2(q+1)}(x+x^q)^6+x^{4(q+1)}(x+x^q)^2+x^{8(q+1)}.
\end{align*}
Recall  that $h=x+x^q$, $r=x^{q+1}$. Thus, we have
\begin{align*}
c^{2(q+1)}(A^qB^3+AB^{3q}+B^{2(q+1)})=&rh^5+rh^7+r^2+r^2h^4+r^2h^6+r^2h^7+r^4h^2+r^4h^4+r^4h^5+r^8,
\end{align*}
and 
\begin{align*}
\frac{H}{D^2}=\frac{r}{h^5}+\frac{r}{h^3}+\frac{r^2}{h^{10}}+\frac{r^2}{h^6}+\frac{r^2}{h^4}+\frac{r^2}{h^3}+\frac{r^4}{h^8}+\frac{r^4}{h^6}+\frac{r^4}{h^5}
+\frac{r^8}{h^{10}}=\mu+\mu^2
\end{align*}
where $\mu=\frac{r+r^4+r^2h+(r+r^2)h^2}{h^5}$.  The rest of this proof is  similar to that of Theorem \ref{vital}, so we omit it here.\end{proof}

\end{document}